\documentclass[11pt]{article}

\usepackage{fullpage,times}
\usepackage{amsthm}

\usepackage{amssymb}
\usepackage{amsmath}
\usepackage{multirow}
\usepackage{threeparttable}
\usepackage{graphicx}
\usepackage{algorithm,algorithmic}
\usepackage{subfigure}
\usepackage{color}
\usepackage{setspace}
\usepackage{url}

\newtheorem{theorem}{Theorem}

\newtheorem{definition}[theorem]{Definition}
\newtheorem{claim}[theorem]{Claim}

\newtheorem{lemma}[theorem]{Lemma}

\newcommand{\mypara}[1]{\smallskip\noindent\textbf{#1.}}  
\newcommand{\mysubsec}[1]{\smallskip\noindent\textbf{\large #1.}}  
\urldef{\magnusmail}\url{mmh@ru.is}
\urldef{\amymail}\url{amywang@hku.hk}
\urldef{\dongxiaomail}\url{dxyu@hust.edu.cn}

\begin{document}

\title{Leveraging Multiple Channels in Ad Hoc Networks\thanks{Partially supported by Iceland Research Foundation grants 120032011 and 152679-051, National Natural Science Foundation
of China grants 61073174, 61373027, 61402461 and HKU Small Project.}}
\author{Magn\'us M. Halld\'orsson\thanks{ICE-TCS, School of Computer Science, Reykjavik University, Iceland; \magnusmail.}
 \and Yuexuan Wang\thanks{College of Computer Science and Technology, Zhejiang University, Hangzhou, 310027, P.R. China.}
\thanks{Department of Computer Science, The University of Hong Kong, Hong Kong, P.R. China; \amymail.}
\and Dongxiao Yu\thanks{Services Computing Technology and System Lab, Cluster and Grid Computing Lab in the School of Computer Science and Technology, Huazhong University of Science and Technology, 1037 Luoyu Road, Wuhan 430074, P.R. China.; \dongxiaomail.}
}
\date{}
\maketitle

\begin{abstract}
We examine the utility of multiple channels of communication in wireless networks under the SINR model of interference. The central question is whether the use of multiple channels can result in linear speedup, up to some fundamental limit. We answer this question affirmatively for the data aggregation problem, perhaps the most fundamental problem in sensor networks. To achieve this, we form a hierarchical structure of independent interest, and illustrate its versatility by obtaining a new algorithm with linear speedup for the node coloring problem.
\end{abstract}

%
%
\setcounter{secnumdepth}{5}
\section{Introduction}

Diversity in wireless networks -- 
having multiple opportunities for communication 
-- is well known to decrease interference, increase reliability, and improve performance~\cite{DGKN12,DGKN11}.
The question is how much it helps and what the limits are to such improvements.
In particular, we seek an answer to the following fundamental question in the context of the SINR model:
\begin{quote}
\emph{
Can we speed distributed wireless algorithms up linearly with the number of channels, up to a fundamental limit?
}
\end{quote}
Thus, we are interested in the fundamental limits of the benefits of diversity.

We focus our attention on data dissemination problems, 
in particular \emph{data aggregation}, sometimes referred to as the ``killer-app'' for sensor networks:
compute a compressible function (e.g., average) of values stored at the nodes~\cite{LHWL10}.

Multiple channels can be available by modulation ranging over frequencies or phases.
They can also be simulated by time-division multiplexing (TDMA) by assigning time slots to the different channels. The converse does not hold, however, as multiple channels are a strictly more constrained form of communication. Namely, whereas nodes can listen (and even choose to send) in all slots of a TDMA schedule, they can only listen on one of the channels. Thus, multiple channels can be viewed as a form of \emph{parallelism} in wireless communication and our inquiry involves the parallelizability of fundamental wireless tasks.

Multiple channels have been found to yield linear speedups in graph-based models, such as for 
broadcast \cite{DGKN11}, minimum dominating sets \cite{DKN12}, leader election~\cite{DGKN12} and maximal independent sets \cite{DGGKN13}. In contrast, essentially the only work on multiple channels in the \emph{signal-to-interference-and-noise ratio} (SINR) model is \cite{YWYYL15},
which attained a sub-linear speedup for local information exchange, but holds only for a restricted number of channels when each message can carry multiple packets. Thus, little has been known about the limits for leveraging multiple channels in an SINR context.

\mypara{Model} We assume synchronized operation with time measured in \emph{rounds}.
Nodes have no power control, no collision detection, but have a carrier sense mechanism in the form of standard signal strength measurements.
The SINR model of interference is assumed, but the parameters ($\alpha, \beta, N$) are allowed to vary within fixed ranges. We assume for simplicity of exposition that nodes are located in the plane, but the results extend to more general metric spaces known as fading metrics.\footnote{A metric space is said to be  \emph{fading} if the path loss exponent $\alpha$ is strictly greater than the doubling dimension of the metric. This is a generalization of the standard requirement of $\alpha>2$ in the two-dimensional Euclidean space, as the two-dimensional Euclidean space has a doubling dimension of $2$. For more details on fading metric, see \cite{H12}.} 
Nodes are given approximate values of SINR parameters and a polynomial bound on the number of nodes,
but have no knowledge of the location of other nodes or their distribution.

\mysubsec{Our Results}
Let $G=(V,E)$ be the communication graph obtained by connecting pairs of nodes that can potentially communicate with each other directly (please refer to Sec.~\ref{sec:model} for detailed definition). Let $D$ be the diameter of $G$, $\Delta$ be its maximum degree, $\mathcal{F}$ be the number of channels, and $n$ the number of nodes (see Sec.~\ref{sec:model} for definitions). We say that an event happens \emph{with high probability} (with respect to $n$), if it happens with probability $1-1/n^c$ for some constant $c > 0$.

We give a randomized algorithm that achieves data aggregation in $O(D+\Delta/\mathcal{F}+\log n\log\log n)$ time with high probability.
Since $\Delta$ is a lower bound for aggregation in single-channel networks, even ones with few hops,
we achieve linear speedup up to the additive $\log n\log\log n$ term. This is essentially best possible for a setting 
where high probability guarantees are required.

Our data aggregation algorithm is based on a data aggregation structure that can be constructed in $O(\log^2n)$ time. If a $\log^{O(1)}n$-approximation of $\Delta$ is known, the time for constructing the aggregation structure is $O(\Delta/\mathcal{F} + \log n \cdot \log\log n)$. Hence, in this case, the total time for accomplishing data aggregation (taking into account the time for structure construction) is $O(D+\Delta/\mathcal{F}+\log n\log\log n)$ with high probability. 

The aggregation structure is of independent interest, as it can be used to solve other core problems.
To illustrate its applicability, we give an algorithm for the node coloring problem that runs in $O(\Delta/\mathcal{F}+\log n\log\log n)$ time with high probability.


\mysubsec{Lower Bounds}
We indicate here briefly why our bounds are close to best possible.
Any global task involving communication requires at least $D$ steps, which yields a lower bound on every instance.
Similarly, $\lceil \log n\rceil$ is a lower bound for data aggregation, since at most half the items can be coalesced in a single round. 
Thus, independent of the parallelization in the form of multiple channels, $\Omega(D + \log n)$ steps are needed.

In a single channel, the term $\Delta$ is necessary for any communication task that involves all nodes when using fixed power assignment such as uniform power.
In particular, consider the ``exponential chain'', where point $i$ is located at position $2^i$ on the real line, $i=1, 2, \ldots, n$.
Then, when using uniform power, at most one successful transmission can occur in a time slot (assuming $\beta \ge 2^{1/\alpha}$) \cite{MW06}.
In particular, aggregation and coloring require $\Delta$ steps in single-channel networks,
and clearly $\mathcal{F}$ channels can reduce the time requirement at most to $\Delta/\mathcal{F}$.
While no proof is known, it is unlikely that power control reduces this bound in the distributed SINR setting; known distributed algorithms all feature time complexity of either the distance diversity (which can be as large as $n$) \cite{BHM13P,HM12P,pei2013distributed}
or terms linear in $\Delta$ \cite{HM11i,HM11,KV10,YHWL12}.





\mysubsec{Related Work}

\mypara{Data Aggregation} 
In single-channel networks, there is a long line of research on data aggregation under different settings in the protocol model~\cite{WHWWJ09,XLMTW11,YLL09} and the SINR model~\cite{ALHN12,BHM13P,DZWWX13,HM12P,HM12S,HWHYL12,LHWL10,LXWTDZQ09}.
Regarding distributed solutions in the SINR model, a
distributed aggregation algorithm with uniform power assignment was proposed in~\cite{LXWTDZQ09}, which achieves a
latency upper bound of $O(D+\Delta)$. Assuming a model where every node in the network knows its position, the network
diameter and the number of neighbors, Li et al.\ \cite{LHWL10} presented a distributed algorithm with a latency bound of
$O(K)$, where $K$ is the logarithm of the ratio between the length of the longest link and that of the shortest link. 
This result additionally needs that nodes can adjust the transmission power arbitrarily. In~\cite{HWHYL12}, Hobbs et al.\ gave a
deterministic algorithm which can accomplish data aggregation in $O(D+\Delta\log n)$ rounds.
An entirely different approach is to use (significant) precomputation to build a fast aggregation structure.
In particular, aggregation can be achieved in optimal $O(D + \log n)$ time \cite{BHM13P,HM12P}, but this uses
$O(K\log^2n)$ time for precomputation and also relies heavily on arbitrary power control.

In multi-channel networks, the multiple-message broadcast algorithm given in \cite{DGGKN13} can be adapted to solve the
data aggregation problem in a graph-based interference model in $O(D+\Delta+\frac{\log^2n}{\mathcal{F}}+\log n\log\log n)$ rounds with high probability, 
but it restricts the number of channels to at most $\log n$. 
An
algorithm for the related broadcast problem was given in ~\cite{DGKN11} for a similar setting but also allowing
disruptions on channels. The work closest to ours is a recent treatment of the local information exchange problem in
multi-channel SINR networks~\cite{YWYYL15}, where Yu et al.\ gave a distributed algorithm attaining a sub-linear
speedup. In the algorithm, the number of channels that can be used effectively is limited to $O(\sqrt{\Delta/\log n})$,
resulting in an $\Omega(\log n\cdot\sqrt{\Delta\log n})$ lower bound on the performance of the algorithm.

\mypara{Coloring} The distributed node coloring problem has been extensively studied since the 1980s as a classical symmetry breaking paradigm \cite{CV86}. Most work has been in message passing models that ignore interference and collisions.
Assuming a graph-based model that defines only direct interference from neighbors, Moscibroda and
Wattenhofer~\cite{MW08} gave an $O(\Delta\log n)$ time randomized algorithm using $O(\Delta)$ colors for
bounded-independence graphs, 
which was lated improved to a $\Delta+1$-coloring in $O(\Delta+\log\Delta\log n)$ time by Schneider and Wattenhofer \cite{SW09}.
Derbel and Talbi \cite{DT10} showed that the algorithm of \cite{MW08} can also
work in the SINR model with the same time and color bounds.
Yu et al.\ \cite{YWHL11A} gave a randomized algorithm with running time $O(\Delta\log n+\log^2n)$ that achieves a
$\Delta+1$-coloring in the SINR model.  All of the above results are for wireless networks with a single channel, and it appears no work has
previously addressed the coloring problem in multiple channel networks, let alone in the SINR model.

\mypara{Backbone Network Construction} Another line of related work is finding dominating sets and/or a
broadcast/ aggregation network in a multi-hop scenario. The work we directly use is that of~\cite{SRS08} with an
algorithm that finds a dominating set in the SINR model in $O(\log n)$ time. 
An algorithm was given in \cite{MW08} that finds a 
maximal independent set running in $O(\log^2n)$ time in the quasi unit disk model, later converted to the SINR
model in~\cite{YWHL11A}. 
Broadcast or aggregation networks among dominators are formed in some works such
as~\cite{BHM13P,HM12P,JKRS13,JK12,JKS13,YHWTL12,YHWYL13}. These works either highly rely on strong assumptions on the connectivity of the network~\cite{YHWTL12,YHWYL13}, use precise location information~\cite{JKRS13,JK12,JKS13}, or arbitrary power
adjustment~\cite{BHM13P,HM12P}. All these works are only for single-channel networks. 


\mysubsec{Roadmap} The formal model, problem definitions and preliminaries are given in Sec.\ \ref{sec:model}. Sec.~\ref{sec:overview} contains a technical overview. In Sec.~\ref{sec:mis}, an algorithm to find ruling sets is introduced, which 
is invoked frequently in the structure construction. The algorithm for constructing the aggregation structure is given in Sec.~\ref{sec:structure} and the data aggregation algorithm in Sec.~\ref{sec:aggreg}. Sec.~\ref{sec:apps} contains the coloring algorithm making use of the aggregation structure.

\section{Model, Problem Formulations and Preliminaries}\label{sec:model}

The network consists of a set $V$ of $n$ nodes with unique IDs that are positioned arbitrarily on a plane.
We focus on the setting of a \emph{uniform} power assignment, where all nodes use the same transmission power $P$.
For two nodes $u$ and $v$, denote by $d(u,v)$ the Euclidean distance between $u$ and $v$.

\mypara{Multiple Communication Channels and Synchronization} Nodes communicate through a shared medium divided
into $\mathcal{F}$ non-overlapping channels. 
Time is divided into synchronized rounds, where each round may contain a constant number of synchronized slots. All nodes start the algorithm at the same time.  In each slot of every round, each node can
select one of the $\mathcal{F}$ channels and either transmit or listen on that channel. A node that
operates on a channel in a given slot learns nothing about events on other channels.

\mypara{Interference and SINR model} Simultaneous transmissions on the same channel interfere with each other. The SINR
model captures the interference by stipulating that a message sent by node $u$ to node $v$ can be correctly received at $v$
iff $(i)$ $u$ and $v$ operate on the same channel and $v$ does not transmit, and $(ii)$ the following signal-to-interference-and-noise-ratio (SINR)
is above a hardware-defined threshold $\beta\ge 1$:
\begin{equation}\label{eq:sinrcondition}
SINR(u,v):=\frac{P/d(u,v)^\alpha}{N+\sum_{w\in S\setminus\{u\}}\frac{P}{d(w,v)^\alpha}}\geq\beta\ ,
\end{equation}
\noindent
where $\alpha > 2$ is the ``path-loss'' constant, $N$ is the ambient noise, and
$S$ is the set of nodes transmitting simultaneously with $u$.

The \emph{transmission range} $R_T$ is the maximum distance at which a transmission can be successfully
decoded (in the absence of other transmissions); by the SINR
condition~(\ref{eq:sinrcondition}), $R_T=(\frac{P}{\beta\cdot N})^{1/\alpha}$.

%

We assume that listening nodes can measure the SINR (only in the case of a successful reception), and the total received power. Nodes can also use this feature to infer (approximate) distances from the sender of a received message. This power reception feature is comparable to the RSSI function of actual wireless motes~\cite{SRS08}. In our algorithm, indeed, it is enough to determine whether the SINR (of a successful reception) or the total received power crosses a single fixed threshold. 

It is always of theoretical interest to determine the tradeoffs between different model assumptions, and to identify the least set of primitives that suffice for efficient execution.
We posit, however, that the default model for wireless algorithms in physical models should feature receiver-side carrier sense ability. Given that such a feature is so standard in even the cheapest hardware and so easily implementable, it would be counterproductive to exclude it. Note that we assume no transmitter-side detection ability.



\mypara{Communication Graph and Notations} 
For parameter $c$, $0 < c < 1$, denote $R_c:=(1-c)R_T$.
 The \emph{communication graph} $G(V, E)$ of a given network consists of
all network nodes and edges $(v, u)$ such that $d(v, u)\leq R_{\epsilon}$, where $0 < \epsilon< 1$ is a fixed model
parameter. Since nodes of distance very close to $R_T$ can only communicate in the absence of other activity in the
network arbitrarily far away, we adopt the standard assumption that a slightly smaller range, $R_\epsilon$, 
is sufficient to communicate \cite{BHM13P,DGKN13,JKRS14}. 

We use standard graph terminology: $N(u)$ is the set of neighbors of node $u$; $d_u = |N(u)|$ is the degree of $u$; and
$\Delta$ is the maximum degree of a node.
The \emph{diameter} $D$ of a graph $G$ is the maximum, over all pairs of nodes $u,v$,
of the shortest hop-distance between $u$ and $v$.

An \emph{$r$-ball} is a disk in the plane of radius $r$.  Denote by $E_v^r$ the $r$-ball centered at node $v$, and overload the
notation to refer also to the set of nodes in the ball.  A node $u$ is an \emph{$r$-neighbor} of (not necessarily
distinct) node $v$ if $d(u,v)\le r$.  An \emph{$r$-dominating set} is a subset $S$ of nodes (called \emph{dominators}) such
that each node in $V$ has an $r$-neighbor in $S$. The \emph{density} of an $r$-dominating set is the maximum
number of dominators in an $r$-ball (over all balls in the plane). A set $S$ of nodes is \emph{$r$-independent}
if no two nodes in $S$ are $r$-neighbors.  An $r$-independent set $S$ is \emph{maximal} if it is also $r$-dominating.

\mypara{Knowledge of Nodes} Nodes know a polynomial approximation to $n$ (i.e., the value of $\log n$, up to constant factors). 
For simplicity of description, we also use $n$ to denote this estimate. We assume that nodes do not know the precise value of the SINR parameters $\alpha$, $\beta$ and $N$ but instead know only upper and lower bounds for the parameters (i.e., $\alpha_{min}$ and $\alpha_{max}$, $\beta_{min}$ and $\beta_{max}$, $N_{min}$ and $N_{max}$).
For simplicity, we perform calculations assuming that exact values of these parameters are known; to deal with uncertainty regarding those parameters, it suffices to choose their maximal/minimal values depending on whether upper or lower estimates are needed. Nodes have no other information, such as the network topology, their neighbors or their location coordinates.

\mypara{Data Aggregation} Initially, each node has a data value.
 The \emph{data aggregation} problem is to compute an aggregate function (e.g., maximum or average) on the input data from all nodes in the network, and inform all nodes of this value as quickly as possible.


\mypara{Preliminaries} The following Chernoff bounds will be used in the analyses of algorithms. The proofs of these bounds can be found in most 
textbooks on probability theory or randomized algorithms.
\begin{lemma}[Chernoff bounds]\label{lem:chernoff}
Let $X_1, X_2, \ldots, X_n$ be independent Bernoulli random variables.
Let $X := \sum^n_{i=1} X_i$ and $\mu := \mathbb{E}[X]$. 
Then, for any $\delta > 0$, it holds that 
\[ Pr[X \geq (1 + \delta)\mu] \leq \left(\frac{e^\delta}{(1 + \delta)^{1 + \delta}}\right)^{\mu}. \]
More precisely, 
\begin{equation}\label{eq:Chernoff1}
Pr[X \geq 2 \mu] \le (e/4)^\mu \le e^{-\mu/3}.
\end{equation}
On the other hand, 
\begin{equation}\label{eq:Chernoff3}
Pr[X \leq \frac{1}{2} \mu] \leq \left(\frac{e^{-1/2}}{(1/2)^{1/2}}\right)^{\mu} = (e/2)^{-\mu/2} \le e^{-\mu/8}\ .
\end{equation}
\end{lemma}

We use a frequently-used argument that shows that well-separated communication can proceed independently.
The proof of this lemma uses the standard technique of bounding interference within concentric circles.

\begin{lemma}\label{le:detertra}
Let $r_1, r_2$ be distance parameters such that $r_2\leq \min\{\left(\frac{\alpha-2}{48\beta(\alpha-1)}\right)^{\frac{1}{\alpha}}\cdot r_1, R_T/2\}$. 
Suppose the set $S_F$ of nodes transmitting on a channel $F$ is $r_1$-independent.
Then, the transmission of each node $v \in S_F$ is received by all $r_2$-neighbors of $v$ that are listening on $F$.
\end{lemma}
\begin{proof}
By assumption, the set $S_F$ satisfies $d(u,v)> r_1$, for any pair of nodes $u,v\in S_F$. 
For a node $w\in S_F$, we compute the interference experienced by a node $x\in Q_F\cap E_w^{r_2}$,
where $Q_F$ is the set of nodes selecting to operate on a channel $F$.
Let $C_t$ be the annulus with
  distance from $w$ in the range $[tr_1,(t+1)r_1)$ for $t\geq 1$. Without confusion, $C_t$ is also used to denote the
  set of nodes in $C_t$ that operate on $F$. Because any two transmitting nodes are separated by $r_1$, an area argument
  implies that $|C_t|\leq 8(2t+1)$. Then we bound the interference at a node $x\in E_w^{r_2}$ caused by other
  transmitters in $S_F$ as follows.
\begin{equation*}
\begin{aligned}
 I_x=\sum_{y\in S_F\setminus\{w\}}\frac{P}{d_{yx}^\alpha}
&\leq\sum_{t=1}^\infty\frac{N\beta R_T^{\alpha}}{(tr_1)^\alpha}\cdot 8(2t+1)\\
&\leq24r_1^{-\alpha}N\beta R_T^{\alpha}\sum_{t=1}^\infty t^{-\alpha+1}\\
&\leq 24r_1^{-\alpha}N\beta R_T^{\alpha}\cdot \frac{\alpha-1}{\alpha-2}\\
&\leq (\frac{R_T^{\alpha}}{r_2^\alpha}-1)N.
\end{aligned}
\end{equation*}
Then by the SINR condition, $x$ can receive the message sent by $u$. 
\end{proof}

Given that each node $u$ transmits with a probability $p_u$, let $P_r(v)=\sum_{u\in E_v^r\cap Q_F}p_u$
be the sum of transmission probabilities of nodes in $E_v^r$ that operate on channel $F$. Using a similar argument as in proving Lemma~\ref{le:detertra} and further considering the transmission probabilities of nodes, 
we can get the following result.
\begin{lemma}\label{le:protran}
Let $R\in \Omega(R_T)$ be a distance, $F$ be a channel and $Q_F$ the set of nodes operating on the channel.
Suppose that each node $u$ transmits on $F$ with probability $p_u$, satisfying
$P_R(v) := \sum_{u\in E_v^R\cap Q_F}p_u \le \psi$.
Then, whenever a node $v$ transmits on $F$, with constant probability $\kappa:=e^{-O((R_T/R)^2\cdot \psi)}\in\Omega(1)$, 
it is heard by all its $R$-neighbors $E_v^R \cap Q_F$ on the channel.
\end{lemma} 

Lemma \ref{le:protran} has been implicitly proved in previous work, such as [11] (Lemma 4.1 and Lemma 4.2). The basic idea of proving Lemma 3 is bounding the interference at the neighbors of a transmitter $v$ from other nearby transmitters (within a specified distance that is a constant times $R$) and faraway transmitters respectively. For the interference from nearby transmitters, based on the facts that these transmitters can be covered by a constant number of $R$-balls and the sum of transmission probabilities of transmitters in each $R$-ball is upper bounded by a constant (as given in the condition), it is easy to show that there are no nearby transmitters with a certain constant probability. For the interference from faraway nodes, it suffices to compute the expected interference at the neighbors of a transmitter $v$ based on the transmission probabilities of nodes, instead of computing the real interference in Eq.~(6). Because the sum of transmission probabilities of each node's neighbors is bounded by a constant (as given in the condition), which means that there are a constant number of transmitters in expectation within the neighborhood of each node, the expected interference at the neighbors of the transmitter $v$ can be bounded by a small constant using the same concentric argument as in Eqt. (6). Then by Markov Inequality, it can shown that with constant probability (determined by $R_T /R$ and $\psi$), the interference at every neighbor of the transmitter $v$ is still upper bounded by a small constant, which is enough to ensure successful receptions. Combining the results for interference bounding from nearby nodes and faraway nodes, Lemma \ref{le:protran} can be proved. For more details, please refer to [11].

In this work, we use Lemma~\ref{le:protran} for only two different distances. So there is a constant lower bound for the probability of successful transmissions. In the subsequence, we still use $\kappa$ to denote this lower bound.

%

\section{Technical Overview}
\label{sec:overview}

Our approach is to build a multi-purpose dissemination structure that we then use in each of our problems.  The structure has
global and local parts, which are linked through the \emph{dominators}, 
the local leaders that collaborate to carry out the global task.

After finding a low-density set of dominators, the other nodes are partitioned into local clusters, each headed by a dominator. These clusters are then colored to disperse the clusters of same color, effectively eliminating interference from other clusters. The clusters are arranged into a communication tree to carry out the global task. These constructions are by now all fairly well known, so we build on previous work, in particular using the $O(\log n)$-round clustering process from \cite{SRS08}.

Our main contributions are in the treatment of the intra-cluster aspects.
We first estimate the size of each cluster, in order to adjust the contention.
We distribute the cluster nodes randomly into channels, and run leader election processes to elect a \emph{reporter} in each channel in $O(\log n)$ rounds. 
We then form a binary tree of $O(\log \mathcal{F})$ levels on the reporters, which is used to aggregate the data to the dominator. 
The total time needed for reporter election and reporter tree construction is $O(\log^2n)$, while the aggregation cost in the clusters is  $O(\Delta/\mathcal{F}+\log n\log\log n)$. If a $\log^{O(1)}n$-approximation of $\Delta$ is known, the reporter election and the reporter tree construction can be done in $O(\Delta/\mathcal{F}+\log n\log\log n)$ time as well.

\section{Ruling Set Algorithm}\label{sec:mis}


We present an algorithm that will be invoked frequently in subsequent sections.
A $(r,s)$-\emph{ruling set} is a subset $S$ of nodes that is both $r$-independent and $s$-dominating.
The algorithm presented finds a $(r,2r)$-ruling set, where $r$ satisfies $r \le R_T/2 = \frac{1}{2}(P/(N\beta))^{1/\alpha}$.

The algorithm has two phases. In the first phase, a constant density $r$-dominating set $X$ is found in $O(\log n)$ rounds using the algorithm of Scheideler et al.~\cite{SRS08}. Let $\mu$ denote an upper bound on the density guaranteed by their algorithm. 
In the remainder of this section, we focus on the second phase, computing a maximal $r$-independent set $S$ among the dominators. 
Namely, $S$ is $r$-independent and each node in $X$ is within distance $r$ from a node in $S$. Then, by the triangular inequality, $S$ forms a $2r$-dominating set of the full set $V$ of nodes.

The strength of signals and interference can yield precious indications about the origin of the signal, and even of interferers.

\begin{definition}
A \emph{clear reception} occurs at a node, for a parameter $r$, if:
a) the message originates from an $r$-neighbor of the node, and
b) the interference sensed is at most 
$T_s=N \cdot \min\{\frac{2^\alpha-1}{2^{\alpha}}, (\frac{1}{2})^{\alpha}\cdot \beta\}$. 
The latter condition ensures that no other $4r$-neighbor transmitted.
\end{definition}
Based on our model assumptions, a node can detect clear receptions. 

The second phase of the algorithm uses three kinds of messages: \textsc{Hello}, \textsc{Ack}, and \textsc{In}.
Let $\gamma=3/(\kappa/2\mu)^2 = 12\mu^2/\kappa^2$, where $\kappa$ is the constant of Lemma~\ref{le:protran}.
The phase consists of $\gamma\ln n$ rounds, each consisting of three slots:
\begin{itemize}
\item\textbf{Slot 1.} Each node transmits \textsc{Hello} independently with probability $1/(2\mu)$.
\item\textbf{Slot 2.} If a node gets a clear reception of \textsc{Hello}, it
sends \textsc{Ack} independently with probability $1/(2\mu)$.
\item\textbf{Slot 3.} If a node sent \textsc{Hello} and received \textsc{Ack} from an $r$-neighbor,
it then joins the set $S$, transmits \textsc{In} and halts.
Otherwise, the node listens; if it receives \textsc{In} from an $r$-neighbor, it halts.
\end{itemize}
If a node is still active after all $\gamma\ln n$ rounds, it then enters the set $S$.
This completes the specification of the second phase, and thus the algorithm.
\medskip

We first argue the correctness of the last step, when dominated nodes bow out.

\begin{lemma}\label{le:MISinform}
If a node joins $S$ in a round,
then all of its (still active) $r$-neighbors halt after that round.
\end{lemma}
\begin{proof}
Let $Y$ be the set of nodes that joined $S$ during the given round, and let $u$ be a node in $Y$.
We claim that all nodes in $E_u^r$ receive \textsc{In} message from $u$.
Let $w$ be a node in $E_u^r$ and observe that the strength of the signal from $u$ received on $w$ is at least $P/r^\alpha$.
Thus, it suffices to show the total interference $I_{Y_u}(w)$ from $Y_u = Y \setminus u$ received 
by $w$ is at most $\frac{1}{\beta}P/r^\alpha - N$.

Let $v$ be the $r$-neighbor of $u$ that sent it \textsc{Ack} and $y$ be a node in $Y_u$.
Since $v$ had a clear reception, $d(v,y) \ge 4r$, while $d(v,w) \le 2r$, since they are both $r$-neighbors of $u$.
Thus, $d(y,w) \ge d(y,v) - d(v,w) \ge \frac{1}{2} d(y,v)$.
Also, the interference $I_{Y_u}(v)$ received by $v$ is then at most $T_s \le \frac{2^\alpha-1}{2^\alpha} N$.
Hence, 
  \[ I_{Y_u}(w) = \sum_{y \in Y_u} \frac{P}{d(y,w)^\alpha} \le 2^\alpha I_{Y_u}(v) \le 2^\alpha T_s \le (2^\alpha-1)N  \le \frac{1}{
\beta} P/r^\alpha - N \ , \]
as desired.
\end{proof}

The main correctness issue is to ensure independence. While the above lemma handles nodes that enter the ruling set
during the main rounds, we use a probabilistic argument to argue that neighbors are unlikely to survive all the rounds
to be able to enter the set $S$ at the end of the execution.

\begin{lemma}\label{le:MISresult}
The algorithm correctly computes a $(r,2r)$-ruling set $S$ in $O(\log n)$ rounds, with high probability.
\end{lemma}

\begin{proof}
By definition of the algorithm, the nodes halt by either joining $S$ or after receiving \textsc{In} from a neighbor. Thus, the solution is an $r$-dominating set of $X$, and hence a $2r$-dominating set of $V$.
It remains to show that $S$ is $r$-independent.

Let $u$, $v$ be nodes in $S$, and suppose without loss of generality that $u$ was added no later than $v$.
If both joined $S$ during the same round 
then they must be of distance at least $3r$ apart (since an $r$-neighbor of $u$ experienced a clear reception).
If $v$ joined $S$ later, it must be more than $r$ away from $u$, since $u$ notified all its $r$-neighbors with an \textsc{In} message, by Lemma \ref{le:MISinform}.
Finally, we show that, with high probability, no $r$-neighbors remain active after all the $\gamma\ln n$ rounds.

Let $u$ and $v$ be $r$-neighbors. 
Observe that the sum of transmission probabilities of nodes in any $r$-ball $E_w^r$ is at most $1/2$ (as the density is at most $\mu$ and each node transmits with probability $1/(2\mu)$).
This allows us to apply Lemma~\ref{le:protran} to determine successful transmissions. 
If $u$ transmits a \textsc{Hello} in a given round, then its neighbors receive it clearly with 
probability is at least $\frac{1}{2\mu}\cdot \kappa$, 
and if a clear reception occurs, then $u$ receives \textsc{Ack}, also with probability at least $\frac{1}{2\mu}\cdot \kappa$.
Hence, if both $u$ and $v$ are active at the beginning of a round, they stay active after that round with probability at most $1 - (\kappa/2\mu)^2$.
Thus, the probability that they stay active for 
all $\gamma \ln n$ rounds is at most $(1-(\kappa/2\mu)^2)^{\gamma\ln n} \le e^{-3\ln n} = n^{-3}$. 
By the union bound, the probability that some $r$-adjacent pairs remains active is at most $n^{-1}$.
\end{proof}

\section{Aggregation Structure Construction}
\label{sec:structure}


We give in this section an algorithm to form a hierarchical aggregation structure.
The algorithm has three parts: forming a dominating set, coloring the dominators to separate them spatially,
and finally forming a tree of reporters to speed up aggregation using the multiple channels.




\subsection{Communication Backbone}\label{sec:clustering}

To reduce computation and communication, we construct an overlay in the form of a connected dominating set.
The dominators function as local leaders of their respective \emph{clusters}, managing the local computation, as well as participating in disseminating the information globally. The dominators are colored to ensure good spatial separation between clusters of same color, which in turn allows the local computation to ignore interference from other clusters.

\subsubsection{Computing a Dominating Set}
%
We first form a \emph{clustering}, which is a function assigning each node a dominator within a specified distance $r$. 

Let $t=\left(\frac{\alpha-2}{48\beta(\alpha-1)}\right)^{1/\alpha}$ and $r_c=\min\{\frac{t}{2t+2}\cdot R_{\epsilon/2}, \frac{\epsilon R_T}{4}\}$. Recall that $R_{\epsilon/2}=(1-\frac{\epsilon}{2})R_T$. We adapt the algorithm of Scheideler et al.~\cite{SRS08} to compute an $r_c$-dominating set of constant density. In that algorithm, a node that receives a message from a dominator becomes a dominatee;
here, we simply additionally require that the node receive a message from a dominator within distance $r_c$. Using the same argument as in~\cite{SRS08}, we have the following result.

\begin{lemma}\label{le:clustering}
There is a distributed algorithm running in time $O(\log n)$ that produces, with high probability, an $r_c$-dominating set of constant density $\mu$, along with the corresponding clustering function. 
\end{lemma}

\subsubsection{Cluster Coloring and a TDMA Scheme of Clusters}\label{sec:tdma}
%
To separate the clusters spatially, we color the dominators so that those within distance $R_{\epsilon/2}$ are assigned different colors, as done by the following algorithm.

The algorithm operates in $\phi$ phases, where $\phi$ is an upper bound on the number of dominators in any disk of radius $R_{\epsilon/2}$. A standard area argument gives an upper bound of 
$\phi := 4\mu(R_{\epsilon/2}+r_c/2)^2/r_c^2\in O(1)$.
In each phase $i$, dominators that are still not colored compute a $(R_{\epsilon/2},R_{\epsilon})$-ruling set, using the
algorithm of Sec.~\ref{sec:mis}, and assign the nodes of the ruling set the color $i$.

The following result follows easily from Lemma~\ref{le:MISresult}.

\begin{lemma}\label{le:coloring}
Given an upper bound $\phi$ on the dominator density,
there is an algorithm for coloring the dominators (assigning $R_{\epsilon/2}$-neighbors different colors)
using $\phi$ colors in $O(\log n)$ rounds.
\end{lemma}
The cluster coloring yields the following TDMA scheme of $\phi$ rounds:
only nodes in clusters of color $i$ transmit in the $i$-th round,  for $i=1, 2, \ldots, \phi$.
A clustering with a proper coloring as described above is called \emph{well-separated}.
Lemma~\ref{le:detertra} and the setting of $r_c$ imply the following result.

\begin{lemma}\label{le:tdma}
If at most one node transmits in each cluster (on a given channel), and only in clusters of a particular color, 
then each such transmission is received by all nodes within the same cluster.
\end{lemma}

Thus, when using the TDMA scheme, communication within clusters can proceed deterministically without concern for outside interference (as long as only one node transmits in a cluster).
For simplicity, in the subsequent sections, we implicitly assume that clusters of the same color execute the algorithm together in the rounds assigned by the TDMA scheme and only consider the algorithm execution of the clusters with a particular color.
This assumption incurs an overhead of only a constant factor $\phi$ on the running time.

\subsection{Reporter Tree Construction in Clusters}\label{sec:reportertreecon}

The tree construction proceeds in three steps. We first estimate the number of nodes in the cluster, which determines the number of channels to which to assign the nodes randomly. Within each channel, a leader known as a reporter is then elected.
Finally, the reporters automatically organize themselves into a complete binary tree, using the channel number as a heap number in the tree.

Denote by $C_v$ the cluster consisting of dominator $v$ and its dominatees.
Denote by $f_v =\min\{\lceil|C_v|/(c_1\log n)\rceil,\mathcal{F}\}$ the number of channels used in cluster $C_v$, where $c_1=24$. 
The setting of $f_v$ ensures (by Chernoff bound) that, with high probability, each channel is assigned at least one node.

The following theorem summarizing the results of this subsection follows from Lemmas~\ref{le:size}, \ref{le:hies}, and \ref{le:treecon} given later.

\begin{theorem}\label{th:tree}
Suppose clusters are well-separated.
There is an algorithm that for each cluster elects a reporter on each of its channels
and organizes them into a complete binary tree, using $O(\log^2 n)$ rounds with high probability.
If a $\log^{O(1)}n$-approximation of $\Delta$ is given, then $O(\Delta/\mathcal{F}+ \log n \cdot \log\log n)$ rounds suffice.
%
\end{theorem}

Since the number of channels used in a cluster depends on its size, we first need to approximate that quantity
and make it known to all dominatees. 

\subsubsection{Cluster Size Approximation}

Suppose an upper bound $\hat{\Delta}$ on the size of any cluster is known.
Consider the following \textbf{Cluster-Size-Approximation (CSA)} problem:
Given a set of nodes partitioned into well-separated clusters, each of size at most $\hat{\Delta}$,
compute a constant approximation of the cluster size and disseminate it to all nodes in the cluster.
In the most general case, $\hat{\Delta}$ can be taken to be $n$.

\paragraph{Cluster Size Approximation with Large $\hat{\Delta}$\vspace{3mm}\\}
The CSA algorithm uses only the first channel, i.e., all nodes operate on a single channel. 
The stage is divided into $\lceil \log \hat{\Delta} \rceil$ phases, each of which contains $\gamma_1\ln n +1$ rounds, where $\gamma_1$ is a constant to be determined.

In all but the last round of a phase, each dominatee $u$ transmits with a specified probability, while the dominators listen.
In rounds of phase $j$, the common transmission probability $p_j$ is $\frac{\lambda}{\hat\Delta} \cdot 2^{j-1}$,
where $\lambda = 1/2$.
Namely, the initial probability is $\lambda/\hat{\Delta}$, and the probability is doubled after each phase.
In the last round of each phase, the dominator sends out a notification if it received enough messages from the nodes in its cluster, in which case all the nodes terminate the algorithm.
If a dominator receives at least $\omega_1 \ln n$ messages in phase $j$ from nodes in its cluster, where $\omega_1 = 36$, 
then it settles for the estimate of $\hat{|C_v|} := \lceil\hat{\Delta}\cdot2^{-j+1}\rceil$ for the number of nodes in its cluster.
Note that, if the contention $P_c(v)$ is constant when the algorithm terminates, then 
$|\hat{C}_v| = \Theta(|C_v|)$, a constant approximation of the true cluster size.


\smallskip

We start with preliminary results before deriving the main result on CSA.

\begin{lemma}
Let $v$ be a dominator and consider a phase of the CSA algorithm.
The following holds with probability $1 - n^{-3}$:
If $P_c(v) < \omega_1/(4\gamma_1)$, then $v$ receives fewer than $\omega_1\ln n$ messages in the phase,
while if $P_c(v) \in (\lambda/2,\lambda]$ and $P_c(w) \le \lambda = 1/2$ for every dominator $w$,
then $v$ receives at least $\omega_1\ln n$ messages.
\label{lem:term}
\end{lemma}

\begin{proof}
Suppose first that $P_c(v) < \omega_1/(4\gamma_1)$.
The dominator $v$ receives a message in a round with probability at most $P_c(v)$, and therefore receives at most 
$\gamma_1 \ln n \cdot \omega_1/(4\gamma_1) = \frac{\omega_1}{4}\ln n$ messages during the phase, in expectation. By Chernoff bound (\ref{eq:Chernoff1}) using $\omega_1=36$, 
it holds with probability $1-n^{-3}$ that $v$ receives at most $\frac{\omega_1}{2}\ln n$ messages.

Suppose now that $P_c(v) \in (\lambda/2, \lambda]$ and that $P_c(w) \le \lambda$ for every dominator $w$.
By Lemma~\ref{le:protran}, if a dominatee transmits in a round, its dominator receives the message with constant probability $\kappa$. 
The probability that $v$ receives some message in a given round of the phase is then at least 
  $\sum_{w\in C_v}p_w\cdot \kappa = \kappa \cdot P_c(v) \geq \kappa\cdot\frac{\lambda}{2}$. 
Then during the first $\gamma_1\ln n$ rounds of phase $j$, $v$ receives at least expected $\kappa\cdot\frac{\lambda}{2}\cdot\gamma_1\ln n$ messages from its dominatees.
Setting $\gamma_1\geq 2\omega_1\cdot \frac{2}{\kappa\lambda}$, 
it follows from Chernoff bound (\ref{eq:Chernoff3}) that $v$ receives at least $\omega_1\ln n$ messages during the first subphase, with probability $1-n^{-3}$, 
in which case it notifies its dominatees to terminate the algorithm.
\end{proof}

\begin{lemma}
With a known upper bound $\hat{\Delta}$ on the maximum cluster size,
the CSA algorithm approximates the size of each cluster within a constant factor in $O(\log \hat{\Delta} \cdot \log n)$ rounds, with high probability. Using the naive bound of $\hat{\Delta} \le n$, the running time is $O(\log^2 n)$.
\label{lem:clustersize-approx1}
\end{lemma}

\begin{proof}
By the first part of Lemma~\ref{lem:term}, using union bounds, it holds with probability at least $1-1/n$ that whenever a dominator $v$ explicitly terminates the algorithm, then $\omega_1/(4\gamma_1)\leq P_c(v)\leq \lambda$. 
Assume that $v$ terminates the algorithm in phase $j$. 
The transmission probability during phase $j$ is $p_j=\frac{1}{2\hat\Delta}\cdot 2^{j-1}$. 
Then, $|C_v|=\frac{P_c(v)}{\frac{1}{2\hat\Delta}\cdot 2^{j-1}}\in [2\hat\Delta\cdot 2^{-j+1}\cdot \omega_1/4\gamma_1, 2\hat\Delta\cdot 2^{-j+1}\cdot \lambda]$. 
Hence, $|\hat{C_v}|=\lceil\hat{\Delta}\cdot2^{-j+1}\rceil\in\Theta(|C_v|)$. In other words,  the estimate $|\hat{C_v}|$ obtained is always a constant approximation of the true cluster size $|C_v|$.
The algorithm is run for at most $\log \hat{\Delta} = O(\log n)$ phases, for a $O(\log^2 n)$ bound on the time complexity.
It remains to argue that the algorithm is explicitly terminated.


By applying the union bounds on the second part of Lemma~\ref{lem:term}, it holds with probability at least $1-1/n$ that $P_c(v) \le \lambda$ is satisfied for every vertex in each phase. 
Initially, $p_1 = \lambda/\hat{\Delta}$, for each dominatee $u$, in which case $P_c(v) \le \lambda$ is satisfied.
If the algorithm operates for all the $\lceil \log \hat{\Delta}\rceil$ phases, then $p_j \ge \lambda/2$ in the last phase $j$, 
in which case $P_c(v) \ge \lambda/2$.
Thus, for each dominator $v$, there is a phase in which the conditions of the second part of Lemma~\ref{lem:term} are satisfied, in which case the dominator terminates the algorithm, with high probability.
\end{proof}

\paragraph{Cluster Size Approximation with Small $\hat{\Delta}$\vspace{3mm}\\}

For the case that $\hat\Delta\leq\mathcal{F}\log^{c}n$, for constant $c\geq1$, the Cluster-Size-Approximation problem
can be solved more efficiently. The basic process is as follows: First, each dominatee selects a
channel uniformly at random. Then, on each channel, the nodes selecting that channel elect a leader and execute the CSA algorithm 
to obtain constant approximation of the number of dominatees in the channel. Finally, the dominator
obtains a constant approximation of the cluster size by polling the estimates from the leaders on each channel, and
sends the estimate to its dominatees on the first channel in the last round. The detailed algorithm and analysis of the following result
are given in the Appendix.

\begin{lemma}
Given knowledge of $\hat{\Delta}$ satisfying $\hat{\Delta} \le \mathcal{F}\log^c n$ for some constant $c\geq 1$,
we can get a constant approximation of the size of each cluster in $O(\log n \cdot \log\log n)$ rounds, with high probability.
\label{lem:low-content-size1}
\end{lemma}

We can combine the two cluster size estimation procedures if $\hat \Delta$ is a $\log^{\hat c}n$-approximation of $\Delta$ for some constant $\hat c\geq 0$:
When $\hat{\Delta}/\mathcal{F} \le \log^{\hat c+2} n$, Lemma \ref{lem:low-content-size1} gives a bound of $O(\log n \cdot \log\log n)$ rounds, while otherwise the bound of Lemma \ref{lem:clustersize-approx1} is $O(\log^2 n) = O(\Delta/\mathcal{F})$ rounds. Hence, based on Lemmas \ref{lem:clustersize-approx1} and \ref{lem:low-content-size1}, we have the following result.

\begin{lemma}\label{le:size}
There is a constant-approximation algorithm for \textbf{Cluster-Size-Approximation} that runs in $O(\log^2 n)$ rounds, with high probability. When given a $\log^{O(1)}n$-approximation of $\Delta$,
there is a constant-approximation algorithm that runs in $O(\Delta/\mathcal{F}+ \log n \cdot \log\log n)$ rounds, with high probability.
\end{lemma}

For simplification, we shall simply use $|C_v|$ to denote the size estimate $|\hat{C_v}|$ derived for cluster $C_v$. Since it is a constant approximation, it will not affect the asymptotic running times.

\subsubsection{Reporter Election and Aggregation Tree Formation}

In this stage, reporters are elected in each cluster simultaneously by running the ruling set algorithm of Sec.~\ref{sec:mis}.
For a cluster $C_v$, a reporter is elected on each of the channels $F_1, F_2, \ldots, F_{f_v}$.
%
To argue correctness, it suffices in light of Lemma \ref{le:MISresult} to show that every channel gets assigned some node.
The expected number of nodes in $C_v$ choosing a channel is $|C_v|/f_v$.
Chernoff bound (\ref{eq:Chernoff3}) and the union bound then imply the desired result with high probability.

\begin{lemma}\label{le:hies}
In each cluster $C_v$, with high probability, exactly one reporter is elected on each of the $f_v$ channels in $O(\log n)$ rounds.
\end{lemma}

We refer to dominatees that are not reporters as \emph{followers}. Thus, $C_v$ is partitioned into one dominator, $f_v$ reporters, and $|C_v|-f_v-1$ followers. In subsequence, we use $X_v=\{u_1,\ldots,u_{f_{v}}\}$ and $Y_v$ to denote the sets of reporters and followers in $C_v$, respectively, where $u_i$ is the reporter elected on channel $F_i$.
%
%
Let $u_0 = v$ refer to the dominator.
We define a complete binary tree rooted at the dominator, with the reporters ordered in level-order, like a binary heap.
Thus, $u_{\lfloor k/2\rfloor}$ is the parent of $u_k$ in the tree, for $k=1, \ldots, f_v$.



Once the reporters are elected, the aggregation tree is then ready to use.
\begin{lemma}\label{le:treecon}
A complete binary tree of $\lfloor \log(f_v+1)\rfloor$ levels is constructed on the reporters for each cluster $C_v$. Operating on well-separated clusters, it can perform a convergecast operation 
deterministically in time $2 \lfloor \log(f_v+1)\rfloor$.
\end{lemma}

\section{Data Aggregation}\label{sec:dataagg}
\label{sec:aggreg}

The data aggregation algorithm consists of three procedures executed in parallel:
The intra-cluster aggregation involves two processes: 
a) collecting the data from followers to the reporters, b) aggregating the data of dominatees using the reporter tree to the dominator,
and finally aggregating the data among the dominators.
The first two procedures can together be referred to as \emph{intra-cluster aggregation}, while the last one is \emph{inter-cluster} aggregation.
In each round there are five slots for these three procedures: a pair of send/acknowledge slots for each of the first two, and
a single slot for the last one.



\mypara{Aggregation from Followers to Reporters}
The execution of this process is divided into phases, each with $\Gamma + 1$ rounds, where $\Gamma := \gamma_2\ln n$ and
$\gamma_2$ is to be determined. 
For a cluster $C_v$, the first $f_v$ channels are used for transmissions. 
The first channel is special in that the dominator listens on it to estimate the contention.
In each phase, the operations of nodes are as follows:

(i) A follower $u\in Y_v$, in each of the first $\Gamma$ rounds, selects one of the first $f_v$ channels uniformly at random, transmits on the selected channel with a specified probability $p_u$ in the first slot, 
and listens in the second slot for an acknowledgement (\emph{ack}) from its reporter. Initially, $p_u$ is set as $p_u=\lambda f_v/|C_v|$ with $\lambda=1/2$. If $u$ receives an ack, it halts.

In the last round, $u$ listens on the first channel. After each phase, if $u$ receives a backoff message from its dominator in the last round, it keeps $p_u$ unchanged, and doubles $p_u$ otherwise. 
 
(ii) A reporter $w\in X_v$ operates on the channel where it is elected. In each of the first $\Gamma$ rounds, $w$ listens in the first slot. If it receives a message from a follower in its cluster, it returns an acknowledgement in the second slot. In the last round, $w$ does nothing.


(iii) The dominator $v$ listens on the first channel during the first $\Gamma$ rounds.
In the last round it transmits a \emph{backoff message} if and only if it heard at least $\Omega := \omega_2\ln n$ messages from followers during the preceding rounds.

In the above algorithm, we set the constant parameters as follows: $\omega_2=96/\kappa_1$ and $\gamma_2=8\omega_2/\kappa_1$, where $\kappa_1\leq 1$ is a constant that will be given in Lemma~\ref{le:dsuc} of the analysis.



\mypara{Aggregation on Reporter Tree}
%
The execution of this process is divided into phases, where each phase contains $\lfloor\log(\mathcal{F}+1)\rfloor-1$ rounds. 


For a cluster $C_v$, the first $f_v$ channels are used for transmissions. As before, use $X_v=\{u_1,\ldots,u_{f_{v}}\}$ to denote the set of reporters and $T_v$ to denote the reporter tree. We enumerate the levels of $T_v$ 
from bottom, i.e., with the leaves at level 1.

In the $s$-th round of a phase, nodes at level $s$ and $s+1$ of $T_v$ execute the algorithm to aggregate from level $s$ to level $s+1$, while other nodes keep silent. 
Each reporter $u_k$ at level $s$  operates on the same channel as its parent, i.e., on the channel $\lfloor k/2\rfloor$. 
If $k$ is odd (even), then $u_k$ transmits its data to its parent $u_{\lfloor k/2 \rfloor}$ in the third (fourth) slot of round $s$, respectively.

\mypara{Inter-cluster Aggregation}
In this procedure, we use a known approach for disseminating data on a constant-density backbone network (e.g., see
Section 5.2 in~\cite{BHM13P}).
The basic idea of the algorithm is to use flooding (with continuous constant-probability transmissions)
to produce an aggregation/broadcast tree, with which data can be aggregated and then broadcast to all nodes in $O(D + \log n)$ rounds with high probability.

\subsection{Analysis}

The main effort of the analysis is on the first procedure, aggregating from the followers to reporters. We address the other two in the final theorem.

To bound the time spent on aggregating from the followers we show that we maintain linear throughput while the contention is high enough. 
Namely, while the contention is above a fixed constant threshold, each reporter makes progress with constant probability,
where progress means aggregating a message from one more follower.
To this end, we show that contention always remains bounded from above, and whenever it becomes low, the transmission probabilities double.
When the contention dips below the threshold, we need only doubly logarithmic number of phases to increase the 
transmission probabilities to constant and aggregate the remaining followers.


The sum of transmission probabilities of followers in a cluster is referred to as the \emph{contention} in the cluster,
and denoted by $P_c(v)=\sum_{u\in Y_v}p_u$.

\begin{definition}[\textbf{Bounded Contention}]\label{pro:transom}
\emph{Bounded Contention} is achieved in a given round if the contention in each cluster is at most half the number of channels alotted, i.e., $P_c(v) \le \lambda f_v = \frac{1}{2}f_v$, for each cluster $C_v$.
\end{definition}

Even if the contention in each cluster is bounded, we cannot directly use the result in Lemma~\ref{le:protran}, as the contention on a particular channel may not be constant bounded. But because followers select the operating channel uniformly at random, it can be seen that the expected contention on each channel can be bounded by $\lambda$. This is enough to  use the interference bounding technique used for proving Lemma~\ref{le:protran}, and we can get the following Lemma~\ref{le:dsuc}. 

We say a follower \emph{succeeds} (to transmit) if its message is properly received by a reporter on a channel. The proof
 detail is omitted because it is very similar to the standard argument given in \cite{GMW08} (Lemma 4.1 and Lemma 4.2).

\begin{lemma}\label{le:dsuc}
Assuming Bounded Contention holds, whenever a follower transmits, it succeeds with probability at least $\kappa_1$, 
for a universal constant $\kappa_1 > 0$.
\end{lemma}

The TDMA scheme ensures that when a follower succeeds, it receives the ack message in the subsequent slot, as argued in Lemma \ref{le:tdma}. 

Using Lemma \ref{le:dsuc}, we can argue the Bounded Contention property.
\begin{lemma}\label{le:dataaggpro}
Bounded Contention holds in every round, with probability  $1 - n^{-1}$.
\end{lemma}
\begin{proof}
  We prove the Lemma by contradiction. Assume that cluster $C_v$ is the first one to violate the Bounded Contention
  property, and that the violating phase is $j$. The initial transmission probability implies that $j>1$. 
  We focus on phase $j-1$.
  In this phase, by assumption, we have $P_c(y)\leq \lambda f_y$ for each cluster $C_y$,
  and since the transmission probability of followers is at most doubled between phases, $P_c(v)\in
  (\lambda f_v/2,\lambda f_v]$.
  The expected number of transmissions by followers in $C_v$ during phase $j-1$ is then at least $\lambda f_v/2\cdot \Gamma$. 
 Since Bounded Contention holds in phase $j-1$, each transmission is successful with some constant probability $\kappa_1$.
  Hence, there are $\lambda/2\cdot \Gamma \kappa_1 = \frac{8}{4} \omega_2 \ln n = 2\Omega$ successful transmissions on each channel, in expectation. 
  Using Chernoff bound (\ref{eq:Chernoff3}), the dominator $v$ receives at least $\Omega$ transmissions with probability $1-n^{-3}$ (as $\omega_2 \ge 12$). 
Then, by Lemma \ref{le:tdma}, $v$ sends a backoff message to all the followers, who keep their transmission probability unchanged after this phase $j-1$.
As a result, the $\lambda f_v$ bound will not be broken in phase $j$, which contradicts with our assumption. So $C_v$ cannot be the first violating cluster with probability $1-n^{-3}$. The Lemma is then proved by the union bound.
\end{proof}


A phase is \emph{increasing} if the transmission probability of the reporters in $C_v$ is doubled after the
phase, i.e., the dominator $v$ receives less than $\omega_2\ln n$ messages, and otherwise it is \emph{unchanging}. 
Let $N_v^j$ denote the total number of transmissions by followers in $C_v$ during phase $j$. A transmission by a
follower $u\in Y_v$ is \emph{successful} if $u$ succeeds in transmitting the data to a reporter.

\begin{lemma}\label{le:unchanging}
Consider a cluster $C_v$.
If a phase $j$ is unchanging, then, with probability at least $1-n^{-3}$,
there are at least $\Omega/4 = \frac{\omega_2}{4} f_v\ln n$ transmissions in the phase, 
of which at least $12 f_v\ln n$ are successful.
\end{lemma}
\begin{proof}
Suppose there are fewer than $\Omega/4$ transmissions in phase $j$.
Then, since channels are chosen with equal probability, the expected number of transmissions in the first channel is at most $\frac{\omega_2}{4}\ln n$.
Then, by Chernoff bound (\ref{eq:Chernoff1}) (since $\omega_2\geq 36$), at most $\frac{\omega_2}{2}\ln n$ transmissions are made in the channel,
with probability $1-n^{-3}$, which implies that the phase is increasing.
Thus, the first part of the lemma holds: if a phase if unchanging, then at least $\Omega/4$ transmissions occur.
By Lemma~\ref{le:dsuc}, the expected number of successful transmissions is then at least $\Omega/4 \cdot \kappa_1 = (\omega_2 f_v\ln n/4)\cdot \kappa_1 = 24 f_v\ln n$. 
Using Chernoff bound (\ref{eq:Chernoff3}), the number of successful transmissions is at least $12 f_v \ln n$, with probability $1-O(n^{-3})$.
\end{proof}

Based on above analysis, we can now get the result for the first procedure.

\begin{lemma}\label{le:intraco}
In each cluster, the data of all followers can be aggregated to the reporters in $O(\frac{\Delta}{\mathcal{F}}+\log n\log\log n)$ rounds, with probability $1-O(n^{-1})$.
\end{lemma}

\begin{proof}
Consider a cluster $C_v$. 
There are at most $O(|C_v|/(12f_v \ln n))= O(1 + \Delta/(\mathcal{F}\log n))$ unchanging phases, 
  by Lemma \ref{le:unchanging}, with probability $1-n^{-3}$.
Also, when the transmission probability of a follower is increased to a constant $\lambda/2$ in a phase, it can successfully send its data to a reporter with probability $\lambda\kappa_1/2$ in each round of the phase by Lemma~\ref{le:dsuc}, and the $\gamma_2\ln n \ge \frac{3\ln n}{\lambda \kappa_1/2}$ rounds in the phase ensure successful transmission with high probability. Hence, there are at most $O(\log(|C_v|/f_v)) = O(\log(\Delta/\mathcal{F}) + \log\log n)$ increasing phases for each cluster, given the initial transmission probability of followers.
Combined, the number of phases is $O(\Delta/(\mathcal{F}\log n) + \log(\Delta/\mathcal{F}) + \log\log n) = O(\Delta/(\mathcal{F}\log n) + \log\log n)$, 
and thus the number of rounds is $O(\Delta/\mathcal{F} + \log n \log\log n)$, with probability $1-n^{-3}$.
The lemma then follows from the union bound over the clusters.
\end{proof}

\begin{theorem}\label{th:da}
Data aggregation can be accomplished in $O(D+\frac{\Delta}{\mathcal{F}}+\log n\log\log n)$ rounds, with high probability.
\end{theorem}

\begin{proof}
We can combine the high probability bounds on each of the three procedures.
By Lemma~\ref{le:intraco}, the aggregation from followers to reporters is achieved 
in $O(\frac{\Delta}{\mathcal{F}}+\log n\log\log n)$ rounds (with probability $1-O(n^{-1})$).
In each cluster $C_v$, the data aggregation from the reporters to the dominator can be accomplished in $O(\log \mathcal{F})$ rounds.
Namely, the construction of the aggregation tree ensures that when a reporter transmits, it is the only one from the same cluster in the same channel, and thus, by Lemma~\ref{le:tdma}, each transmission is successful. The number of rounds to aggregate from reporters to dominator then equals the height of the tree, or $\lfloor\log(\mathcal{F}+1)\rfloor$.
Finally, Theorem 3 in~\cite{BHM13P} achieves (the inter-cluster) aggregation on the dominators in $O(D + \log n)$ rounds.
\end{proof}

\section{Coloring}\label{sec:apps}

Using the aggregation structure, the data of dominatees can be efficiently aggregated to a dominator, 
as shown in Sec.~\ref{sec:dataagg}.
This aggregation structure can be used to solve fundamental problems other than data aggregation, which we illustrate on the node coloring problem.

\mypara{Algorithm} In the constructed aggregation structure, the dominators are colored with cluster colors $1,2,\ldots,\phi$ for some constant $\phi$ such that dominators within distance $R_{\epsilon/2}$ receive different cluster colors (refer to Sec. \ref{sec:tdma}).
We then allocate to each dominator of cluster color $i$ the sequence of colors $k\phi+i: k=0,1,2,\ldots$ to assign to its cluster nodes.

Operating on each cluster $C_v$, the algorithm consists of four procedures: 

\begin{enumerate}
\item The followers execute the data aggregation algorithm of Sec.~\ref{sec:aggreg} to send their IDs to the reporters, by which each reporter will acquire the knowledge of all of its followers. An aggregation tree on all nodes in $C_v$ is then constructed based on the reporter tree by adding links that connect each reporter and the followers following it.

\item Each reporter forwards the number of nodes in its subtree (including the reporters and the followers) to its parent in the reporter tree.

\item The color range (the range of $k$, which determines the set of available colors) of each reporter and its followers is then disseminated to each reporter via the reporter tree. In particular, on the reporter tree, each node $u$ (recall that the root is the dominator) determines the color ranges of its two children based on the color range assigned to $u$ and the number of nodes in the subtree of its children. The distribution of the color range uses an inverse process to the aggregation on the reporter tree given in Sec.~\ref{sec:aggreg}. 

\item For a reporter $u$, let $B_u$ denote the set of colors assigned to it (which can be derived using the color range assigned to $u$). Each reporter $u$ then assigns a different color in $B_u$ to each of its followers and announces the color assignment one by one to its followers.
\end{enumerate}

Because the first procedure uses a randomized algorithm, and the other three procedures are done by letting nodes execute the deterministic TDMA scheme given in Section 5.1, to avoid the interference between the executions of different procedures among clusters,  we run procedures in separate slots of each round. Specifically, in each round, there are four slots for the execution of each of the four procedures.


\mypara{Analysis}
\begin{lemma}\label{le:intracloringtime}
For each cluster $C_v$, after $O(\frac{\Delta}{\mathcal{F}}+\log n\log\log n)$ rounds, each node in $C_v$ will get a different color with high probability. And the total number of colors used is $O(\Delta)$.
\end{lemma}
\begin{proof}
As for the time complexity, 
the first procedure takes $O(\Delta/\mathcal{F}+\log n\log\log n)$ rounds, with high probability, 
by Lemma~\ref{le:intraco};
the second and third procedures take $O(\log \mathcal{F})$ rounds, or proportional to the height of the tree (using Lemma~\ref{le:treecon}); 
and finally, the fourth procedure takes as many rounds as a reporter has followers, or $O(|C_v|/f_v + \log n)$,
since these messages are successfully received because the reporters transmit on different channels.
Because nodes have a constant approximation of $|C_v|$, and knowledge of $n$ (a polynomial estimate) and the number of channels $\mathcal{F}$, 
they can each determine the completion time of each procedure.

The time bound can be obtained by the execution time of each procedure. We next show that each node in $C_v$ gets a different color and the total number of colors used is $O(\Delta)$. 

By Lemma~\ref{le:intraco}, each follower can send its ID to a reporter with high probability. We claim that each follower transmits its ID to only one reporter. This follows from Lemma~\ref{le:tdma}. By this Lemma, once a follower transmits a message to a reporter, it will receive an ack message in the same round. Hence, the sets of followers of reporters are disjoint. With this claim, we can see that the aggregation tree on all nodes in $C_v$ is correctly constructed in the first procedure, i.e., every node is in the tree and has exactly one parent. Then in the second procedure, each reporter will get the exact number of nodes in its subtree by the analysis in Theorem~\ref{th:da}. Based on this knowledge and because the aggregation tree is correctly constructed, after the third procedure, reporters will get disjoint color ranges and the number of colors used is $|C_v|\in O(\Delta)$. Hence, after the fourth procedure, each node will get a different color.
\end{proof}
\begin{theorem}\label{th:nodecoloring}
A proper coloring with $O(\Delta)$ colors can be computed in $O(\frac{\Delta}{\mathcal{F}}+\log n\log\log n)$ rounds with high probability. 
\end{theorem}

\begin{proof}
The total time used for the coloring is given in Lemma~\ref{le:intracloringtime}. By the algorithm, it is easy to see that the total number of colors used is $\phi\cdot O(\Delta)\in O(\Delta)$.

We next show the correctness of the coloring algorithm. For any two neighboring nodes $u,v$ that are in different clusters, their dominators have distance at most $\epsilon R_T/4+R_{\epsilon}+\epsilon R_T/4=R_{\epsilon/2}$. By the algorithm, the color sets given to the clusters in which $u,v$ stay are disjoint. Hence, $u,v$ will not get the same color. For any pair of neighboring nodes in the same cluster, they will also be assigned different colors by Lemma~\ref{le:intracloringtime}.
\end{proof}


\begin{thebibliography}{10}

\bibitem{ALHN12}
M.~K. An, N.~X. Lam, D.~T. Huynh, and T.~N. Nguyen.
\newblock Minimum latency data aggregation in the physical interference model.
\newblock {\em Computer Communications}, 35(18):2175--2186, 2012.

\bibitem{BHM13P}
M.~H. Bodlaender, M.~M. Halld\'{o}rsson, and P.~Mitra.
\newblock {Connectivity and Aggregation in Multihop Wireless Networks}.
\newblock In {\em PODC}, pages 355--364. ACM, 2013.

\bibitem{CV86}
R.~Cole and U.~Vishkin.
\newblock {Deterministic Coin Tossing with Applications to Optimal Parallel
  List Ranking}.
\newblock {\em Inf. Control}, 70(1):32--53, July 1986.

\bibitem{DGGKN13}
S.~Daum, M.~Ghaffari, S.~Gilbert, F.~Kuhn, and C.~Newport.
\newblock {Maximal independent sets in multichannel radio networks}.
\newblock In {\em PODC}, pages 335--344, 2013.

\bibitem{DGKN12}
S.~Daum, S.~Gilbert, F.~Kuhn, and C.~Newport.
\newblock {Leader election in shared spectrum radio networks}.
\newblock In {\em PODC}, pages 215--224, 2012.

\bibitem{DGKN13}
S.~Daum, S.~Gilbert, F.~Kuhn, and C.~Newport.
\newblock {Broadcast in the Ad Hoc SINR Model}.
\newblock In {\em DISC}, pages 358--372, 2013.

\bibitem{DKN12}
S.~Daum, F.~Kuhn, and C.~Newport.
\newblock {Efficient symmetry breaking in multi-channel radio networks}.
\newblock In {\em DISC}, pages 238--252, 2012.

\bibitem{DT10}
B.~Derbel and E.~Talbi.
\newblock {Distributed Node Coloring in the SINR Model}.
\newblock In {\em ICDCS}, pages 708--717, 2010.

\bibitem{DGKN11}
S.~Dolev, S.~Gilbert, M.~Khabbazian, and C.~Newport.
\newblock {Leveraging channel diversity to gain efficiency and robustness for
  wireless broadcast}.
\newblock In {\em DISC}, pages 252--267, 2011.

\bibitem{DZWWX13}
H.~Du, Z.~Zhang, W.~Wu, L.~Wu, and K.~Xing.
\newblock Constant-approximation for optimal data aggregation with physical
  interference.
\newblock {\em Journal of Global Optimization}, 56(4):1653--1666, 2013.

\bibitem{GMW08}
O.~Goussevskaia, T.~Moscibroda, and R.~Wattenhofer.
\newblock {Local broadcasting in the physical interference model}.
\newblock In {\em DIALM-POMC '08}, pages 35--44, 2008.

\bibitem{H12}
M.~M. Halld\'{o}rsson.
\newblock Wireless scheduling with power control.
\newblock {\em ACM Trans. Algorithms}, 9(1):7:1--7:20, Dec. 2012.

\bibitem{HM11i}
M.~M. Halld\'{o}rsson and P.~Mitra.
\newblock {Nearly Optimal Bounds for Distributed Wireless Scheduling in the
  {SINR} Model}.
\newblock In {\em ICALP}, pages 625--636, 2011.

\bibitem{HM12P}
M.~M. Halld\'{o}rsson and P.~Mitra.
\newblock {Distributed Connectivity of Wireless Networks}.
\newblock In {\em PODC}, 2012.

\bibitem{HM11}
M.~M. Halld\'{o}rsson and P.~Mitra.
\newblock {Towards tight bounds for local broadcasting}.
\newblock In {\em FOMC}, 2012.

\bibitem{HM12S}
M.~M. Halld\'{o}rsson and P.~Mitra.
\newblock Wireless connectivity and capacity.
\newblock In {\em SODA}, pages 516--526. SIAM, 2012.

\bibitem{HWHYL12}
N.~Hobbs, Y.~Wang, Q.-S. Hua, D.~Yu, and F.~Lau.
\newblock {Deterministic distributed data aggregation under the SINR model}.
\newblock In {\em TAMC'12}, pages 385--399, 2012.

\bibitem{JKRS13}
T.~Jurdzinski, D.~Kowalski, M.~Rozanski, and G.~Stachowiak.
\newblock Distributed randomized broadcasting in wireless networks under the
  {SINR} model.
\newblock In {\em DISC}, pages 373--387. 2013.

\bibitem{JKRS14}
T.~Jurdzinski, D.~Kowalski, M.~Rozanski, and G.~Stachowiak.
\newblock On the impact of geometry on ad hoc communication in wireless
  networks.
\newblock In {\em PODC '14}, 2014.

\bibitem{JK12}
T.~Jurdzinski and D.~R. Kowalski.
\newblock Distributed backbone structure for algorithms in the {SINR} model of
  wireless networks.
\newblock In {\em DISC}, pages 106--120, 2012.

\bibitem{JKS13}
T.~Jurdzinski, D.~R. Kowalski, and G.~Stachowiak.
\newblock {Distributed Deterministic Broadcasting in Uniform-power Ad Hoc
  Wireless Networks}.
\newblock In {\em FCT'13}, pages 195--209, 2013.

\bibitem{KV10}
T.~Kesselheim and B.~V\"{o}cking.
\newblock {Distributed contention resolution in wireless networks}.
\newblock In {\em DISC}, pages 163--178, 2010.

\bibitem{LHWL10}
H.~Li, Q.-S. Hua, C.~Wu, and F.~Lau.
\newblock {Minimum-latency Aggregation Scheduling in Wireless Sensor Networks
  Under Physical Interference Model}.
\newblock In {\em MSWiM '10}, pages 360--367, 2010.

\bibitem{LXWTDZQ09}
X.-Y. Li, X.~Xu, S.~Wang, S.~Tang, G.~Dai, J.~Zhao, and Y.~Qi.
\newblock Efficient data aggregation in multi-hop wireless sensor networks
  under physical interference model.
\newblock In {\em MASS '09}, pages 353--362, 2009.

\bibitem{MW06}
T.~Moscibroda and R.~Wattenhofer.
\newblock The complexity of connectivity in wireless networks.
\newblock In {\em INFOCOM}, pages 1--13, April 2006.

\bibitem{MW08}
T.~Moscibroda and R.~Wattenhofer.
\newblock {Coloring unstructured radio networks}.
\newblock {\em Distributed Computing}, 21(4):271--284, 2008.

\bibitem{pei2013distributed}
G.~Pei and A.~K.~S. Vullikanti.
\newblock Distributed approximation algorithms for maximum link scheduling and
  local broadcasting in the physical interference model.
\newblock In {\em INFOCOM}, 2013.

\bibitem{SRS08}
C.~Scheideler, A.~Richa, and P.~Santi.
\newblock {An $O(\log N)$ Dominating Set Protocol for Wireless Ad-hoc Networks
  Under the Physical Interference Model}.
\newblock In {\em MobiHoc '08}, pages 91--100, 2008.

\bibitem{SW09}
J.~Schneider and R.~Wattenhofer.
\newblock {Coloring Unstructured Wireless Multi-hop Networks}.
\newblock In {\em PODC}, pages 210--219, 2009.

\bibitem{WHWWJ09}
P.-J. Wan, S.~C.-H. Huang, L.~Wang, Z.~Wan, and X.~Jia.
\newblock Minimum-latency aggregation scheduling in multihop wireless networks.
\newblock In {\em MobiHoc '09}, pages 185--194, 2009.

\bibitem{XLMTW11}
X.~Xu, M.~Li, X.~Mao, S.~Tang, and S.~Wang.
\newblock A delay-efficient algorithm for data aggregation in multihop wireless
  sensor networks.
\newblock {\em Parallel and Distributed Systems, IEEE Transactions on},
  22(1):163--175, Jan 2011.

\bibitem{YLL09}
B.~Yu, J.~Li, and Y.~Li.
\newblock Distributed data aggregation scheduling in wireless sensor networks.
\newblock In {\em INFOCOM'09}, pages 2159--2167, 2009.

\bibitem{YHWL12}
D.~Yu, Q.-S. Hua, Y.~Wang, and F.~Lau.
\newblock {An $O(\log n)$ Distributed Approximation Algorithm for Local
  Broadcasting in Unstructured Wireless Networks}.
\newblock In {\em DCOSS '12}, pages 132--139, 2012.

\bibitem{YHWTL12}
D.~Yu, Q.-S. Hua, Y.~Wang, H.~Tan, and F.~Lau.
\newblock {Distributed multiple-message broadcast in wireless ad-hoc networks
  under the SINR model}.
\newblock In {\em SIROCCO'12}, pages 111--122, 2012.

\bibitem{YHWYL13}
D.~Yu, Q.-S. Hua, Y.~Wang, J.~Yu, and F.~Lau.
\newblock {Efficient distributed multiple-message broadcasting in unstructured
  wireless networks}.
\newblock In {\em INFOCOM'13}, pages 2427--2435, 2013.

\bibitem{YWHL11A}
D.~Yu, Y.~Wang, Q.-S. Hua, and F.~Lau.
\newblock {Distributed ($\Delta$+1)-coloring in the physical model}.
\newblock In {\em ALGOSENSORS'11}, pages 145--160, 2011.

\bibitem{YWYYL15}
D.~Yu, Y.~Wang, Y.~Yan, J.~Yu, and F.~Lau.
\newblock Speedup of information exchange using multiple channels in wireless
  ad hoc networks.
\newblock In {\em INFOCOM'15}, 2015.

\end{thebibliography}

\begin{appendix}

\section{Cluster Size Approximation with Small $\hat{\Delta}$}

When the contention is known to be small relative to the number of channels, we can reduce the time complexity for computing the cluster size. Here we consider the case that $\hat\Delta\leq\mathcal{F}\log^{c}n$ for some constant $c\geq1$.

\textbf{Algorithm.} For each cluster $C_v$, the algorithm consists of four procedures:

1. Initially, each dominatee in $C_v$ selects a channel from $\mathcal{F}$ uniformly at random. 
On each channel, the nodes selecting the channel elect a leader by executing the ruling-set algorithm given in Sec.~\ref{sec:mis}. 
This procedure consists of $\gamma_3\ln n$ rounds, where $\gamma_3$ is set to be a sufficiently large constant such that there are enough rounds for the execution of the algorithm in Sec.~\ref{sec:mis}. 

2. On each channel, nodes execute the CSA Algorithm with $\hat\Delta = \gamma_3\ln^cn$, where the leader functions as the dominator on the channel. 

3. The leaders aggregate the number of nodes that selected the channels they dominate. This procedure consists of $O(\log \mathcal{F})$ rounds. In particular, denote by $U_v=\{x_1,\ldots,x_\mathcal{F}\}$ the set of leaders in cluster $C_v$. Note that there may be some channels without nodes selecting it and thus without leaders elected on them. Hence, there may be some nodes $x_i$ missing. For each channel that does not have nodes, we add an auxiliary node, and it will be introduced how to deal with these auxiliary nodes in the aggregation process.

We first construct a binary tree on these $\mathcal{F}$ nodes rooted at the dominator using the same manner as the reporter tree construction in Sec.~\ref{sec:reportertreecon}. Then we use the data aggregation algorithm on the reporter tree given in Sec.~\ref{sec:aggreg} to aggregate the number of nodes to the dominator. Specifically, we need to handle here the auxiliary nodes. The solution is to divide each slot in each round into two sub-slots (recall that there are two slots in each round for the data aggregation on reporter trees), and make a parent send the ack message when it receives a message from its children. For each node $x_j$ transmits, if it does not receive the ack message from its parent, which means that its parent is an auxiliary node, $x_j$ will function as its parent in the subsequent aggregation process.


4. Finally, in a single round, $v$ broadcasts the estimate of the cluster size to its dominatees on the first channel. 


\textbf{Analysis.} 


\textbf{Proof of Lemma~\ref{lem:low-content-size1}.} Consider a cluster $C_v$. We analyze the four procedures one by one. We first bound the number of nodes operating on each channel in the first procedure.
\begin{claim}
For a cluster $C_v$, in the first procedure, there are at most $2\ln^c n$ nodes on each channel with probability $1-n^{-2}$.
\end{claim}
\begin{proof}
Because dominatees select channels uniformly at random, the expected number of dominatees selecting each channel is at most $\ln^cn$. Consider a channel $F$. Using Chernoff bound (\ref{eq:Chernoff3}), we get that the number of dominatees selecting $F$ is at most twice the expectation, with probability $1-n^{-3}$. By the union bound on all channels, the result follows. 
\end{proof}

A channel $F$ is \emph{nonempty} with respect to a cluster $C_v$ if there are dominatees in $C_v$ selecting it in the first procedure. Using a similar argument for proving Lemma~\ref{le:hies}, we have the following result for the first procedure.

\begin{claim}\label{le:fir}
For each cluster and each nonempty channel $F$, exactly one leader is elected on $F$ in $O(\log n)$ rounds,
with probability $1-n^{-2}$.
\end{claim}

Using a similar argument for proving Lemma~\ref{lem:clustersize-approx1}, we have the following result for the second procedure.
\begin{claim}\label{le:sec}
Each leader in each cluster
can get an absolute constant approximation of the number of dominatees selecting its channel in $O(\log n\log\log n)$ rounds, with probability $1-n^{-2}$. 
\end{claim}

By Lemma~\ref{le:tdma}, a node will receive an ack message after it sends a message to its parent if its parent is not an auxiliary node. Hence, the auxiliary nodes will not affect the aggregation process in the third procedure. Hence, we have the following result.
\begin{claim}\label{le:thi}
For a cluster $C_v$, the estimates of leaders will be aggregated to the dominator $v$ in $O(\log \mathcal{F})$ rounds.
\end{claim}

After the estimates of leaders are aggregated to the dominator, the dominator $v$ will get a constant approximation of the cluster size by Claim~\ref{le:sec}. Then in the fourth procedure, $v$ can send the estimate of the cluster size to all dominatees by Lemma~\ref{le:tdma}. Adding the time used in each procedure, each node in cluster $C_v$ will get a constant approximation of the cluster size in $O(\log n\log\log n)$ rounds with probability $1-O(n^{-2})$. The result is then proved by the union bound.\qed

\end{appendix}

\end{document}